\documentclass{article}

\usepackage[pdftex,paper=a4paper,left=3cm,right=3cm,top=4cm,bottom=4cm]{geometry}
\usepackage{amsmath, amssymb, amsthm}
\usepackage{graphicx}
\usepackage[pdftex,linktocpage,plainpages=false,breaklinks]{hyperref}

\allowdisplaybreaks[4]

\newcommand{\R}{\mathbb{R}}

\renewcommand{\S}{\mathcal{S}}
\renewcommand{\P}{\mathcal{P}}

\newcommand{\SET}[1]{\left\{#1\right\}}

\newcommand{\DOT}{\,.}
\newcommand{\KOMMA}{\,,}
\newcommand{\WHERE}{\,\colon\,}
\newcommand{\ceil}[1]{\left\lceil #1 \right\rceil}
\newcommand{\floor}[1]{\left\lfloor #1 \right\rfloor}

\renewcommand{\l}{\lambda}
\newcommand{\ra}{\rightarrow}

\newcommand{\OMEGA}[1]{\Omega\left(#1\right)}
\newcommand{\THETA}[1]{\Theta\left(#1\right)}
\newcommand{\LOG}[2][]{\log_{#1}\left(#2\right)}

\newcommand{\MIN}[1]{\min\SET{#1}}
\newcommand{\LIST}[3][1]{#2_{#1}, \ldots, #2_{#3}}
\newcommand{\SLIST}[3][1]{#2 {#1}, \ldots, #2 {#3}}
\newcommand{\NLIST}[2][1]{\SLIST[#1]{}{#2}}

\renewcommand{\d}{{d_h}}
\newcommand{\dt}{{d_h}}

\newtheorem{theorem}{Theorem}
\newtheorem{lemma}[theorem]{Lemma}
\newtheorem{corollary}[theorem]{Corollary}

\title{Lower Bounds for the Smoothed Number of\\ Pareto optimal Solutions}

\author{Tobias Brunsch\vspace{-1.0mm}\\
		\footnotesize Department of Computer Science\vspace{-1.2mm}\\
		\footnotesize University of Bonn, Germany\vspace{-1.2mm}\\
        \footnotesize \texttt{brunsch@cs.uni-bonn.de}
   \and
        Heiko R\"oglin\vspace{-1.0mm}\\
		\footnotesize Department of Computer Science\vspace{-1.2mm}\\
		\footnotesize University of Bonn, Germany\vspace{-1.2mm}\\
        \footnotesize \texttt{heiko@roeglin.org}}

\begin{document}

\maketitle

\begin{abstract}
In 2009, R\"{o}glin and Teng showed that the smoothed number of Pareto optimal
solutions of linear multi-criteria optimization problems is polynomially bounded
in the number~$n$ of variables and the maximum density~$\phi$ of the semi-random
input model for any fixed number of objective functions. Their bound is, however,
not very practical because the exponents grow exponentially in the number~$d+1$ of
objective functions. In a recent breakthrough, Moitra and O'Donnell improved
this bound significantly to $O \big( n^{2d} \phi^{d(d+1)/2} \big)$.

An ``intriguing problem'', which Moitra and O'Donnell formulate in their
paper, is how much further this bound can be improved. The previous lower
bounds do not exclude the possibility of a polynomial upper bound whose
degree does not depend on~$d$. In this paper we
resolve this question by constructing a class of instances with $\Omega ( (n \phi)^{(d-\LOG{d}) \cdot (1-\THETA{1/\phi})} )$ Pareto
optimal solutions in expectation. For the bi-criteria case we present a higher lower
bound of $\Omega ( n^2 \phi^{1 - \THETA{1/\phi}} )$,
which almost matches the known upper bound of $O( n^2 \phi)$.
\end{abstract}

\section{Introduction}

In multi-criteria optimization problems we are given several objectives and aim at finding a solution that is simultaneously optimal in all of them. In most cases the objectives are conflicting and no such solution exists. The most popular way to deal with this problem is based on the following simple observation. If a solution is \emph{dominated} by another solution, i.e.\ it is worse than the other solution in at least one objective and not better in the others, then this solution does not have to be considered for our optimization problem. All solutions that are not dominated are called \emph{Pareto optimal}, and the set of these solutions is called \emph{Pareto set}.

\paragraph{Knapsack Problem with Groups}

Let us consider a variant of the \emph{knapsack problem} which we call \emph{restricted multi-profit knapsack problem}. Here, we have~$n$ objects $\LIST{a}{n}$, each with a weight~$w_i$ and a profit vector~$p_i \in \R^d$ for a positive integer~$d$. By a vector~$s \in \SET{ 0, 1 }^n$ we can describe which object to put into the knapsack. In this variant of the knapsack problem we are additionally given a set~$\S \subseteq \SET{ 0, 1 }^n$ of \emph{solutions} describing all combinations of objects that are allowed. We want to simultaneously minimize the total weight and maximize all total profits of a solution~$s$. Thus, our optimization problem, denoted by $K_\S(\{ \LIST{a}{n} \})$, can be written as
\begin{itemize}
\item[] \textbf{minimize} $\sum \limits_{i=1}^n w_i \cdot s_i$, \text{\,\, and \,\,} \textbf{maximize} $\sum \limits_{i=1}^n (p_i)_j \cdot s_i$ for all~$j = \NLIST{d}$
\item[] \textbf{subject to}~$s$ in the feasible region~$\S$.
\end{itemize}
For~$\S = \SET{ 0, 1 }^n$ we just write $K(\{ \LIST{a}{n} \})$ instead of $K_\S(\{ \LIST{a}{n} \})$.

In this paper we will consider a special case of the optimization problem above where we partition the objects into groups. For each group a set of allowed subgroups is given. Independently of the choice of the objects outside this group we have to decide for one of those subgroups. Hence, the set~$\S$ of solutions is of the form $\S = \prod_{i=1}^k \S_i$ where~$\S_i \subseteq \SET{ 0, 1 }^{n_i}$ and $\sum_{i=1}^k n_i = n$. We refer to this problem as \emph{multi-profit knapsack problem with groups}.

\paragraph{Smoothed Analysis}

For many multi-criteria optimization problems the worst-case size of the Pareto
set is exponential in the number of variables. However, worst-case analysis is
often too pessimistic, whereas average-case analysis assumes a certain
distribution on the input universe, which is usually unknown. \emph{Smoothed
analysis}, introduced by Spielman and Teng~\cite{DBLP:journals/jacm/SpielmanT04}
to explain the efficiency of the simplex algorithm in practice despite its
exponential worst-case running time, is a combination of both approaches.
Like in a worst-case analysis the model of smoothed analysis still considers adverserial
instances. In contrast to the worst-case model, however, these instances are
subsequently slightly perturbed at random, for example by Gaussian noise.
This assumption is made to model that often the input an algorithm gets is subject to imprecise
measurements, rounding errors, or numerical imprecision. 
In a more general model of smoothed analysis, introduced by Beier and
V\"{o}cking~\cite{DBLP:journals/jcss/BeierV04}, the adversary is even allowed to
specify the probability distribution of the random noise. The influence he can
exert is described by a parameter~$\phi$ denoting the maximum density of the
noise.

For the restricted multi-profit knapsack problem we use the following smoothing
model which has also been used by Beier and
V\"{o}cking~\cite{DBLP:journals/jcss/BeierV04}, by Beier, R\"{o}glin, and
V\"{o}cking~\cite{DBLP:conf/ipco/BeierRV07}, by R\"{o}glin and
Teng~\cite{DBLP:conf/focs/RoglinT09}, and by Moitra and
O'Donnell~\cite{MoitraO10}. Given positive integers~$n$ and~$d$ and a real~$\phi
\geq 1$ the adversary can specify a set~$\S \subseteq \SET{ 0, 1 }^n$ of
solutions, arbitrary object weights $\LIST{w}{n}$ and density functions $f_{i,j}
\colon [-1, 1] \ra \R$ such that $f_{i,j} \leq \phi$, $i = \NLIST{n},\ j =
\NLIST{d}$. Now the profits~$(p_i)_j$ are drawn independently according to the
density functions~$f_{i,j}$. The \emph{smoothed number of Pareto optimal
solutions} is the largest expected size of the Pareto set of $K_\S(\{ \LIST{a}{n} \})$
that the adversary can achieve by choosing the set $\S$, the weights $w_i$, and the
probability densities $f_{i,j}\leq\phi$ for the profits~$(p_i)_j$.

\paragraph{Previous Work}

Beier and V\"{o}cking~\cite{DBLP:journals/jcss/BeierV04} showed that for $d = 1$
the expected size of the Pareto set is $O(n^4 \phi)$. Furthermore, they showed a
lower bound of $\OMEGA{n^2}$ if all profits are uniformly drawn from $[0, 1]$.
Later, Beier, R\"{o}glin, and V\"{o}cking~\cite{DBLP:conf/ipco/BeierRV07}
improved the upper bound to $O(n^2 \phi)$ by analyzing the so-called \emph{loser
gap}. R\"{o}glin and Teng~\cite{DBLP:conf/focs/RoglinT09} generalized the notion
of this gap to higher dimensions, i.e.\ $d \geq 2$, and gave the first
polynomial bound in $n$ and $\phi$ for the smoothed number of Pareto optimal
solutions. Furthermore, they were able to bound higher moments. The degree of
the polynomial, however, was $d^{\THETA{d}}$. Recently, Moitra and
O'Donnell~\cite{MoitraO10} showed a bound of $O(n^{2d} \phi^{d(d+1)/2})$,
which is the first polynomial bound for the expected size of the Pareto
set with degree polynomial in $d$. 
An ``intriguing problem'' with which Moitra and O'Donnell conclude their paper is
whether their upper bound could be significantly improved, for example to $f(d, \phi) n^2$.
Moitra and O'Donnell suspect that for constant $\phi$ there should be a lower
bound of $\OMEGA{n^d}$. In this paper we resolve this question almost completely.

\paragraph{Our Contribution}

For $d = 1$ we prove a lower bound $\OMEGA{\MIN{n^2 \phi^{1-\THETA{1/\phi}},
2^{\THETA{n}}}}$. This is the first bound with dependence on $n$ and $\phi$ and
it nearly matches the upper bound $O(\MIN{n^2 \phi, 2^n})$. For $d \geq 2$ we
prove a lower bound $\OMEGA{ \MIN{ (n\phi)^{(d-\LOG[2]{d}) \cdot (1 -
\THETA{1/\phi})}, 2^{\THETA{n}} } }$. This is the first bound for the general
multi-criteria case. Still, there is a significant gap between this lower bound
and the upper bound of $O(\MIN{n^{2d} \phi^{d(d+1)/2},2^n})$  shown by Moitra and O'Donnell, but the exponent of $n$ is
nearly $d - \LOG[2]{d}$. Hence our lower bound is close to the lower bound of $\OMEGA{n^d}$
conjectured by Moitra and O'Donnell.

\section{The Bi-criteria Case}

In this section we present a lower bound for the expected number of Pareto optimal solutions in bi-criteria optimization problems that shows that the upper bound of Beier, R\"{o}glin, and V\"{o}cking \cite{DBLP:conf/ipco/BeierRV07} cannot be significantly improved. To prove this lower bound, we consider a class of instances for a variant of the knapsack problem, in which subsets of items can form groups such that either all items in a group have to be put into the knapsack or none of them.
\begin{theorem}
\label{bi.mainthm}
There is a class of instances for the bi-criteria knapsack problem with groups for which the expected number of Pareto-optimal solutions is lower bounded by
\[ \OMEGA{ \MIN{ n^2 \phi^{1-\Theta(1/\phi)}, 2^{\THETA{n}} } } \KOMMA \]
where~$n$ is the number of objects and~$\phi$ is the maximum density of the profits' probability distributions.
\end{theorem}
Note, that Beier, R\"{o}glin, and V\"{o}cking \cite{DBLP:conf/ipco/BeierRV07} proved an upper bound of $O(\MIN{ n^2 \phi, 2^n })$. That is, the exponents of~$n$ and~$\phi$ in the lower and the upper bound are asymptotically the same.

For our construction we use the following lower bound from Beier and V\"{o}cking.
\begin{theorem}[\cite{DBLP:journals/jcss/BeierV04}]
\label{bi.thm.n.squared}
Let $\LIST{a}{n}$ be objects with weights $\SLIST{2^}{n}$ and profits $\LIST{p}{n}$ that are independently and uniformly distributed in~$[0, 1]$. Then, the expected number of Pareto optimal solutions of $K(\{ \LIST{a}{n} \})$ is $\OMEGA{n^2}$.
\end{theorem}
Note that scaling all profits does not change the Pareto set and hence Theorem~\ref{bi.thm.n.squared} remains true if the profits are chosen uniformly from~$[0, a]$ for an arbitrary~$a > 0$. We will exploit this observation later in our construction.

The idea how to create a large Pareto set is what we call the copy step. Let us assume we have an additional object~$b$ with weight~$2^{n+1}$ and fixed profit~$q$. The solutions from $K(\{ \LIST{a}{n}, b \})$ can be considered as solutions from $K(\{ \LIST{a}{n} \})$ that do not use object~$b$ or as solutions from $K(\{ \LIST{a}{n} \})$ that additionally use object~$b$. By the choice of the weight of~$b$, a Pareto optimal solution from $K(\{ \LIST{a}{n} \})$ is also a Pareto optimal solution from $K(\{ \LIST{a}{n}, b\})$ as object~$b$ alone is heavier than all objects $\LIST{a}{n}$ together. The crucial observation is that a solution that uses object~$b$ is Pareto optimal if and only if its profit is larger than the largest profit of any Pareto optimal solution from $K(\{ \LIST{a}{n} \})$ and if it is Pareto optimal for $K(\{ \LIST{a}{n} \})$ when not using~$b$. The first condition is always fulfilled if we choose the profit~$q$ large enough. In this case we can view the Pareto optimal solutions using object~$b$ as copies of the Pareto optimal solutions that do not use~$b$.
\begin{lemma}
\label{bi.lemma.copy}
Let $\LIST{a}{n}$ be objects with weights $\SLIST{2^}{n}$ and profits $\LIST{p}{n} \geq 0$ and let~$b$ be an object with weight~$2^{n+1}$ and profit~$q > \sum_{i=1}^n p_i$. Furthermore, let~$\P$ denote the Pareto set of $K(\{ \LIST{a}{n} \})$ and let~$\P'$ denote the Pareto set of $K(\{ \LIST{a}{n}, b\})$. Then,~$\P'$ is the disjoint union of $\P'_0 := \SET{ (s, 0) \WHERE s \in \P }$ and $\P'_1 := \SET{ (s, 1) \WHERE s \in \P }$ and thus $|\P'| = 2\cdot|\P|$.
\end{lemma}
Figure~\ref{bi.fig.copy.step} visualizes the proof idea. If we represent all solutions by a weight-profit pair in the weight-profit space, then the set of solutions using object~$b$ is the set of solutions that do not use object~$b$, but shifted by~$(2^{n+1}, q)$. As both components of this vector are chosen sufficiently large, there is no domination between solutions from different copies and hence the Pareto optimal solutions of $K(\{ \LIST{a}{n}, b \})$ are just the copies of the Pareto optimal solutions of $K(\{ \LIST{a}{n} \})$.
\begin{figure}[t]
  \begin{center}
    \includegraphics[width=0.5\textwidth]{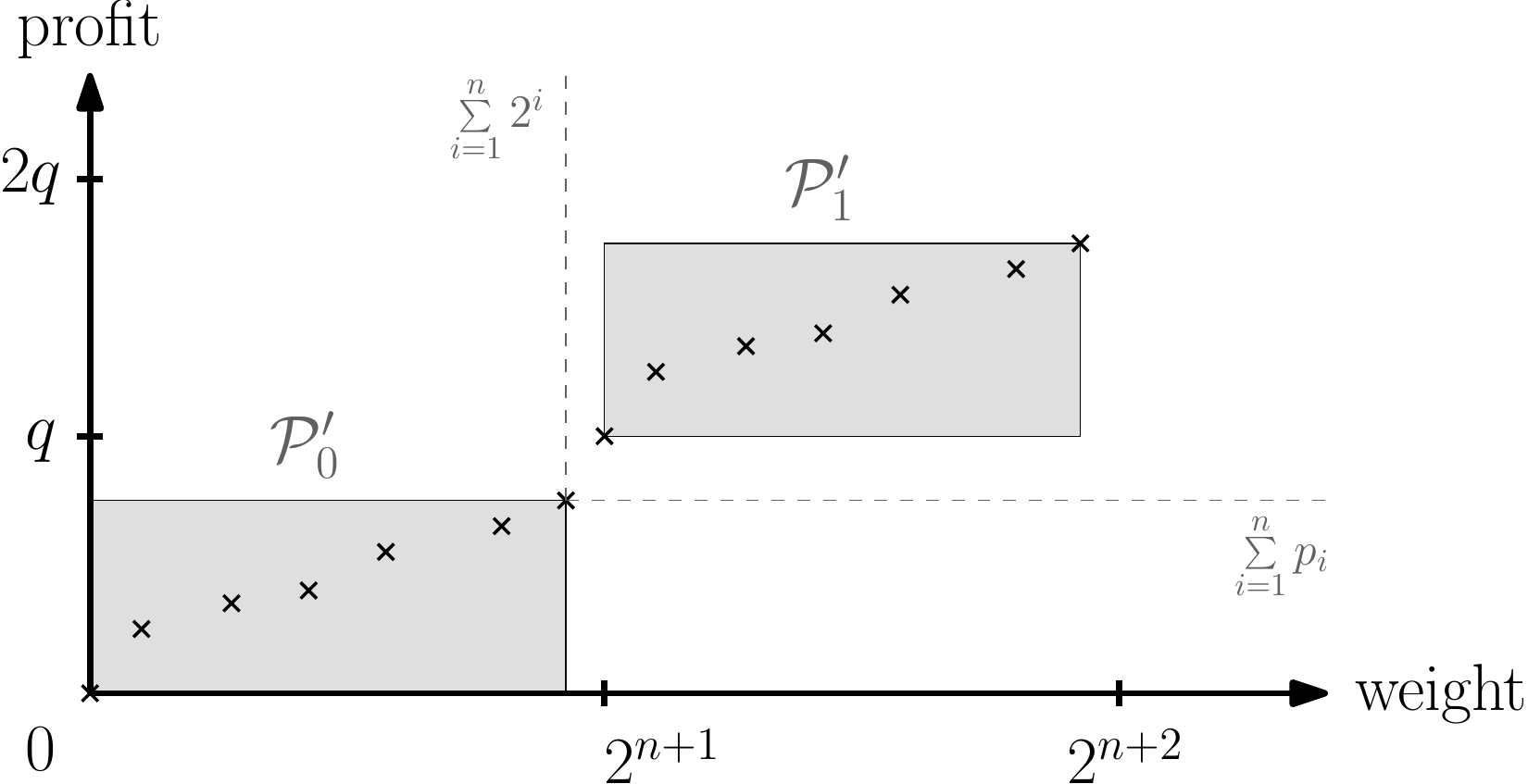}
  \end{center}
  \caption{The copy step. The Pareto set $\P'$ consist of two copies of the Pareto set $\P$.}
  \label{bi.fig.copy.step}
\end{figure}
Now we use the copy idea to construct a large Pareto set. Let $\LIST{a}{n_p}$ be objects with weights $\SLIST{2^}{n_p}$ and with profits $\LIST{p}{n_p} \in P := [0, \frac{1}{\phi}]$ where~$\phi > 1$, and let $\LIST{b}{n_q}$ be objects with weights $\SLIST[n_p+1]{2^}{n_p+n_q}$ and with profits
\[ q_i \in Q_i := \left( m_i - \frac{\ceil{m_i}}{\phi}, m_i \right] \KOMMA \mbox{ where } m_i = \frac{n_p+1}{\phi-1} \cdot \left( \frac{2\phi-1}{\phi-1} \right)^{i-1} \DOT \]
To apply Lemma~\ref{bi.lemma.copy}, we first have to show that we chose the intervals~$Q_i$ appropriately. Additionally, we implicitely show that the lower boundaries of the intervals~$Q_i$ are non-negative.
\begin{lemma}
\label{bi.lemma.appr.Qi}
Let $\LIST{p}{n_p} \in P$ and let $q_i \in Q_i$. Then, $q_i > \sum_{j=1}^{n_p} p_j + \sum_{j=1}^{i-1} q_j$ for all $i = \NLIST{n_q}$.
\end{lemma}
\begin{proof}
Using the definition of~$m_i$, we get
\begin{align*}
q_i &> m_i - \frac{\ceil{m_i}}{\phi} \geq m_i - \frac{m_i+1}{\phi} = \frac{\phi-1}{\phi} \cdot m_i - \frac{1}{\phi} = \frac{n_p+1}{\phi} \cdot \left( \frac{2\phi-1}{\phi-1} \right)^{i-1} - \frac{1}{\phi}
\end{align*}
and
\begin{align*}
\sum \limits_{j=1}^{n_p} p_j + \sum \limits_{j=1}^{i-1} q_j &\leq \sum \limits_{j=1}^{n_p} \frac{1}{\phi} + \sum \limits_{j=1}^{i-1} m_j = \frac{n_p}{\phi} + \sum \limits_{j=1}^{i-1} \frac{n_p+1}{\phi-1} \cdot \left( \frac{2\phi-1}{\phi-1} \right)^{j-1} \cr
&= \frac{n_p}{\phi} + \frac{n_p+1}{\phi-1} \cdot \frac{\left( \frac{2\phi-1}{\phi-1} \right)^{i-1} - 1}{\frac{2\phi-1}{\phi-1} - 1} = \frac{n_p}{\phi} + \frac{n_p+1}{\phi} \cdot \left( \left( \frac{2\phi-1}{\phi-1} \right)^{i-1} - 1 \right) \cr
&= \frac{n_p+1}{\phi} \cdot \left( \frac{2\phi-1}{\phi-1} \right)^{i-1} - \frac{1}{\phi} \DOT \qedhere
\end{align*}
\end{proof}
Combining Theorem~\ref{bi.thm.n.squared}, Lemma~\ref{bi.lemma.copy} and Lemma~\ref{bi.lemma.appr.Qi}, we immediately get a lower bound for the knapsack problem using the objects $\LIST{a}{n_p}$ and $\LIST{b}{n_q}$ with profits chosen from~$P$ and~$Q_i$, respectively.
\begin{corollary}
\label{bi.corol.pareto}
Let $\LIST{a}{n_p}$ and $\LIST{b}{n_q}$ be as above, but the profits~$p_i$ are chosen uniformly from~$P$ and the profits~$q_i$ are arbitrarily chosen from~$Q_i$. Then, the expected number of Pareto optimal solutions of $K(\{ \LIST{a}{n_p}, \LIST{b}{n_q} \})$ is $\OMEGA{n_p^2 \cdot 2^{n_q}}$.
\end{corollary}
\begin{proof}
Because of Lemma~\ref{bi.lemma.appr.Qi}, we can apply Lemma~\ref{bi.lemma.copy} for each realization of the profits $\LIST{p}{n_p}$ and $\LIST{q}{n_q}$. This implies that the expected number of Pareto optimal solutions is~$2^{n_q}$ times the expected size of the Pareto set of $K(\{ \LIST{a}{n_p} \})$ which is $\OMEGA{n_p^2}$ according to Theorem~\ref{bi.thm.n.squared}.
\end{proof}
The profits of the objects~$b_i$ grow exponentially and leave the interval~$[0, 1]$. We resolve this problem by splitting each object~$b_i$ into~$k_i := \ceil{m_i}$ objects $b_i^{(1)},  \ldots, b_i^{(k_i)}$ with the same total weight and the same total profit, i.e.\ with weight~$2^{n_p+i}/k_i$ and profit
\[ q_i^{(l)} \in Q_i/k_i := \left( \frac{m_i}{k_i} - \frac{1}{\phi}, \frac{m_i}{k_i} \right] \DOT \]
As the intervals~$Q_i$ are subsets of~$\R_+$, the intervals~$Q_i/k_i$ are subsets of~$[0, 1]$. It remains to ensure that for any fixed~$i$ all objects~$b_i^{(l)}$ are treated as a group. This can be done by restricting the set~$\S$ of solutions. Let $\S_i = \SET{ (0, \ldots, 0), (1, \ldots, 1) } \subseteq \SET{ 0, 1 }^{k_i}$. Then, the set~$\S$ of solutions is defined as
\[ \S := \SET{ 0, 1 }^{n_p} \times \prod \limits_{i=1}^{n_q} \S_i \DOT \]
By choosing the set of solutions that way, the objects $b_i^{(1)}, \ldots, b_i^{(k_i)}$ can be viewed as substitute for object~$b_i$. Thus, a direct consequence of Corollary~\ref{bi.corol.pareto} is the following.
\begin{corollary}
\label{bi.corol.final.pareto}
Let~$\S$, $\LIST{a}{n_p}$ and $b_i^{(l)}$ be as above, and let the profits $\LIST{p}{n_p}$ be chosen uniformly from~$P$ and let the profits $q_i^{(1)}, \ldots, q_i^{(k_i)}$ be chosen uniformly from~$Q_i / k_i$. Then, the expected number of Pareto optimal solutions of $K_\S (\{ \LIST{a}{n_p} \} \cup \{ b_i^{(l)} \WHERE i = \NLIST{n_q}, \ l = \NLIST{k_i} \})$ is $\OMEGA{n_p^2 \cdot 2^{n_q}}$.
\end{corollary}
The remainder contains just some technical details. First, we give an upper bound for the number of objects~$b_i^{(l)}$.
\begin{lemma}
\label{bi.lemma.count.objects}
The number of objects~$b_i^{(l)}$ is upper bounded by $n_q + \frac{n_p+1}{\phi} \cdot \left( \frac{2\phi-1}{\phi-1} \right)^{n_q}$.
\end{lemma}
\begin{proof}
The number of objects~$b_i^{(l)}$ is $\sum_{i=1}^{n_q} k_i = \sum_{i=1}^{n_q} \ceil{m_i} \leq n_q + \sum_{i=1}^{n_q} m_i$, and
\[ \sum \limits_{i=1}^{n_q} m_i = \frac{n_p+1}{\phi-1} \cdot \sum \limits_{i=1}^{n_q} \left( \frac{2\phi-1}{\phi-1} \right)^{i-1} \leq \frac{n_p+1}{\phi-1} \cdot \frac{\left( \frac{2\phi-1}{\phi-1} \right)^{n_q}}{\frac{2\phi-1}{\phi-1} - 1} = \frac{n_p+1}{\phi} \cdot \left( \frac{2\phi-1}{\phi-1} \right)^{n_q} \DOT \qedhere \]
\end{proof}
Now we are able to prove Theorem~\ref{bi.mainthm}.
\begin{proof}[Proof of Theorem~\ref{bi.mainthm}]
Without loss of generality let~$n \geq 4$ and $\phi \geq \frac{3+\sqrt{5}}{2} \approx 2.62$. For the moment let us assume $\phi \leq (\frac{2\phi-1}{\phi-1})^\frac{n-1}{3}$. This is the interesting case leading to the first term in the minimum in Theorem~\ref{bi.mainthm}. We set $\hat{n}_q := \frac{\LOG{\phi}}{\LOG{ \frac{2\phi-1}{\phi-1} }} \in [ 1, \frac{n-1}{3}]$ and $\hat{n}_p := \frac{n-1-\hat{n}_q}{2} \geq \frac{n-1}{3} \geq 1$. All inequalities hold because of the bounds on~$n$ and~$\phi$. We obtain the numbers~$n_p$ and~$n_q$ by rounding, i.e.\ $n_p := \floor{\hat{n}_p} \geq 1$ and $n_q := \floor{\hat{n}_q} \geq 1$. Now we consider objects $\LIST{a}{n_p}$ with weights~$2^i$ and profits chosen uniformly from~$P$, and objects~$b_i^{(l)}$, $i = \NLIST{n_q}$, $l = \NLIST{k_i}$, with weights~$2^{n_p+i}/k_i$ and profits chosen uniformly from~$Q_i/k_i$. Observe that~$P$ and all~$Q_i/k_i$ have length~$\frac{1}{\phi}$ and thus the densities of all profits are bounded by~$\phi$. Let~$N$ be the number of all these objects. By Lemma~\ref{bi.lemma.count.objects}, this number is bounded by
\begin{align*}
N &\leq n_p + n_q + \frac{n_p+1}{\phi} \cdot \left( \frac{2\phi-1}{\phi-1} \right)^{n_q} \leq \hat{n}_p + \hat{n}_q + \frac{\hat{n}_p+1}{\phi} \cdot \left( \frac{2\phi-1}{\phi-1} \right)^{\hat{n}_q} \cr
&= \hat{n}_p + \hat{n}_q + \frac{\hat{n}_p+1}{\phi} \cdot \phi = 2\hat{n}_p + \hat{n}_q + 1 = n \DOT
\end{align*}
Hence, the number~$N$ of binary variables we actually use is at most~$n$, as required. As set of solutions we consider $\S := \SET{ 0, 1 }^{n_p} \times \prod \limits_{i=1}^{n_q} \S_i$. Due to Corollary~\ref{bi.corol.final.pareto}, the expected size of the Pareto set of $K_\S(\{ \LIST{a}{n_p}\} \cup \{ b_i^{(l)} \WHERE i = \NLIST{n_q}, \ l = \NLIST{k_i} \})$ is
\begin{align*}
\OMEGA{n_p^2 \cdot 2^{n_q}} &= \OMEGA{\hat{n}_p^2 \cdot 2^{\hat{n}_q}} = \OMEGA{ \hat{n}_p^2 \cdot 2^{\frac{\LOG{\phi}}{\LOG{ \frac{2\phi-1}{\phi-1} }}} } = \OMEGA{ n^2 \cdot \phi^{\frac{\LOG{2}}{\LOG{ \frac{2\phi-1}{\phi-1} }}} } \cr
&= \OMEGA{ n^2 \cdot \phi^{1-\Theta{1/\phi}} } \KOMMA
\end{align*}
where the last step holds because
\[ \frac{1}{\LOG[2]{ 2+\frac{c_1}{\phi-c_2} }} = 1 - \frac{\LOG{ 1+\frac{c_1}{2\phi-2c_2} }}{\LOG{ 2 + \frac{c_1}{\phi-c_2} }} = 1 - \frac{\THETA{ \frac{c_1}{2\phi-2c_2} }}{\THETA{1}} = 1 - \THETA{ \frac{1}{\phi} } \]
for any constants $c_1, c_2 > 0$. We formulated this argument slightly more general than necessary as we will use it again in the multi-criteria case.

In the case $\phi > (\frac{2\phi-1}{\phi-1})^\frac{n-1}{3}$ we construct the same instance as above, but for maximum density~$\phi' > 1$ where $\phi' = (\frac{2\phi'-1}{\phi'-1})^{\frac{n-1}{3}}$. Since~$n \geq 4$, the value~$\phi'$ exists, is unique and $\phi' \in \left[ \frac{3+\sqrt{5}}{2}, \phi \right)$. As above, the expected size of the Pareto set is
\begin{align*}
\OMEGA{ n^2 \cdot 2^{\frac{\LOG{\phi'}}{\LOG{ \frac{2\phi'-1}{\phi'-1} }}} } &= \OMEGA{ n^2 \cdot 2^{\frac{n-1}{3}} } = \OMEGA{ n^2 \cdot 2^{\THETA{n}} } = \OMEGA{ 2^{\THETA{n}} } \DOT \qedhere
\end{align*}
\end{proof}

\section{The Multi-criteria Case}

In this section we present a lower bound for the expected number of Pareto optimal solutions in multi-criteria optimization problems. For this, we construct a class of instances for a variant of the knapsack problem where each object has one weight and~$d$ profits and where objects can form groups. We restrict our attention to~$d \geq 2$ as we discussed the case~$d = 1$ in the previous section.

\begin{theorem}
\label{multi.mainthm}
For any fixed integer~$d \geq 2$ there is a class of instances for the $(d+1)$-dimensional knapsack problem with groups for which the expected number of Pareto-optimal solutions is lower bounded by
\[ \OMEGA{ \MIN{ (n\phi)^{(d-\LOG{d}) \cdot (1 - \THETA{1/\phi})}, 2^{\THETA{n}} } } \KOMMA \]
where~$n$ is the number of objects and~$\phi$ is the maximum density of the profit's probability distributions.
\end{theorem}

Unfortunately, Theorem~\ref{multi.mainthm} does not generalize
Theorem~\ref{bi.mainthm}. This is due to the fact that, though we know an
explicit formula for the expected number of Pareto optimal solutions if all
profits are uniformly chosen from~$[0, 1]$, we were not able to find a simple
non-trivial lower bound for it. Hence, in the general multi-criteria case, we
concentrate on analyzing the copy and split steps.

In the bi-criteria case we used an additional object~$b$ to copy the Pareto set
(see Figure~\ref{bi.fig.copy.step}). For that we had to ensure that every
solution using this object has higher weight than all solutions without~$b$. The
opposite had to hold for the profit. Since all profits are in~$[0,1]$, the
profit of every solution must be in~$[0,n]$. As the Pareto set of the first $n_p \leq n/2$
objects has profits in $[0,n/(2\phi)]$, we could fit $n_q = \THETA{\LOG{\phi}}$ copies of 
this initial Pareto set into the interval~$[0,n]$.

In the multi-criteria case, every solution has a profit in~$[0,n]^d$.
In our construction, the initial Pareto set consists only of a single solution, but
we benefit from the fact that the number of mutually non-dominating copies of the
initial Pareto set that we can fit into the hypercube~$[0,n]^d$ grows quickly with~$d$.

Let us consider the case that we have some Pareto set~$\P$ whose profits lie in some hypercube~$[0,a]^d$.
We will create $\binom{d}{\dt}$ copies of this Pareto set; one for every vector~$x \in \SET{ 0, 1 }^d$
with exactly $\dt = \ceil{d/2}$ ones. Let $x \in \SET{ 0, 1 }^d$ be such a vector. Then we generate the
corresponding copy~$C_x$ of the Pareto set~$\P$ by shifting it by $a + \varepsilon$ in every dimension~$i$ with~$x_i=1$.
If all solutions in these copies have higher weights than the solutions in the initial Pareto set~$\P$,
then the initial Pareto set stays Pareto optimal. Furthermore, for each pair of copies~$C_x$ and~$C_y$,
there is one index~$i$ with~$x_i=1$ and~$y_i=0$. Hence, solutions from~$C_y$ cannot dominate solutions from~$C_x$.
Similarly, one can argue that no solution in the initial copy can dominate any solution from~$C_x$.
This shows that all solutions in copy~$C_x$ are Pareto optimal. All the copies (including the initial one)
have profits in $[0,2a + \varepsilon]^d$ and together $|\P|\cdot \big( 1+\binom{d}{\dt} \big) \geq |\P| \cdot 2^d/d$ solutions.

We start with an initial Pareto set of a single solution with profit in~$[0,1/\phi]^d$, and hence we can
make $\THETA{\LOG{n\phi}}$ copy steps before the hypercube~$[0,n]^d$ is filled. In each of these steps
the number of Pareto optimal solutions increases by a factor of at least~$2^d/d$, yielding a total number
of at least
\[ \left( \frac{2^d}{d} \right)^{\THETA{\LOG{n \phi}}} = (n \phi)^{\THETA{d-\LOG{d}}} \] 
Pareto optimal solutions.

In the following, we describe how these copy steps can be realized in the
restricted multi-profit knapsack problem. Again, we have to make a split step because the profit of every object must be in~$[0,1]^d$. Due to such technicalities, the actual bound we prove looks slightly different than the one above.
It turns out that we need (before splitting)~$d$ new objects
$\LIST{b}{d}$ for each copy step in contrast to the bi-criteria case, where (before splitting) a single
object~$b$ was enough. 

Let~$n_q \geq 1$ be an arbitrary positive integer and let~$\phi \geq 2d$ be a
real. We consider objects~$b_{i,j}$ with weights~$2^i/\dt$ and profit vectors
\[ q_{i,j} \in Q_{i,j} := \prod \limits_{k=1}^{j-1} \left[ 0, \frac{\ceil{m_i}}{\phi} \right] \times \left( m_i - \frac{\ceil{m_i}}{\phi}, m_i \right] \times \prod \limits_{k=j+1}^d \left[ 0, \frac{\ceil{m_i}}{\phi} \right] \KOMMA \]
where~$m_i$ is recursively defined as
\begin{equation}
\label{eq.multi.recurrence}
m_0 := 0 \ \mbox{ and } \ m_i := \frac{1}{\phi - d} \cdot \left( \sum \limits_{l=0}^{i-1} \left( m_l \cdot \left( \phi + d \right) + d \right) \right),\ i = \NLIST{n_q} \DOT
\end{equation}
The explicit formula for this recurrence is
\begin{equation*}
\label{eq.multi.explicit}
m_i = \frac{d}{\phi+d} \cdot \left( \left( \frac{2\phi}{\phi-d} \right)^i - 1 \right), \ i = \NLIST{n_q} \DOT
\end{equation*}
The $d$-dimensional interval~$Q_{i,j}$ is of the form that the $j^\text{th}$~profit of object~$b_{i,j}$ is large and all the other profits are small as discussed in the motivation.

Let~$H(x)$ be the \emph{Hamming weight} of a $0$-$1$-vector~$x$, i.e.\ the number of ones in~$x$, and let $\hat{\S} := \{ x \in \SET{ 0, 1 }^d \WHERE H(x) \in \SET{ 0, \d } \}$ denote the set of all $0$-$1$-vectors of length~$d$ with~$0$ or~$\d$ ones. As set~$\S$ of solutions we consider $\S := \hat{\S}^{n_q}$.
\begin{lemma}
\label{multi.lemma.copy}
Let the set~$\S$ of solutions and the objects~$b_{i,j}$ be as above. Then, each solution~$s \in \S$ is Pareto optimal for $K_\S(\{ b_{i,j} \WHERE i = \NLIST{n_q}, \ j = \NLIST{d} \})$.
\end{lemma}
\begin{proof}
We show the statement by induction over~$n_q$ and discuss the base case and the inductive step simultaneously because of similar arguments. Let $\S' := \hat{\S}^{n_q-1}$ and let $(s, s_{n_q}) \in \S' \times \hat{\S}$ be an arbitrary solution from~$\S$. Note that for~$n_q = 1$ we get~$s = \l$, the $0$-$1$-vector of length $0$. First we show that there is no domination within one copy, i.e.\ there is no solution of type $(s', s_{n_q}) \in \S$ that dominates $(s, s_{n_q})$. For~$n_q = 1$ this is obviously true. For~$n_q \geq 2$ the existence of such a solution would imply that~$s'$ dominates~$s$ in the knapsack problem $K_{\S'}(\{ b_{i,j} \WHERE i = \NLIST{n_q - 1}, \ j = \NLIST{d} \})$. This contradicts the inductive hypothesis.

Now we prove that there is no domination between solutions from different copies, i.e.\ there is no solution of type $(s', s'_{n_q}) \in \S$ with $s'_{n_q} \neq s_{n_q}$ that dominates $(s, s_{n_q})$. If $s_{n_q} = \vec{0}$, then the total weight of the solution $(s, s_{n_q})$ is at most $\sum_{i=1}^{n_q-1} 2^i < 2^{n_q}$. The right side of this inequality is a lower bound for the weight of solution $(s', s'_{n_q})$ because $s'_{n_q} \neq s_{n_q}$. Hence, $(s', s'_{n_q})$ does not dominate $(s, s_{n_q})$. Finally, let us consider the case $s_{n_q} \neq \vec{0}$. There must be an index~$j \in [d]$ where $(s_{n_q})_j = 1$ and $(s'_{n_q})_j = 0$. We show that the $j^\text{th}$~total profit of $(s, s_{n_q})$ is higher than the $j^\text{th}$~profit of $(s', s'_{n_q})$. The former one is strictly bounded from below by $m_{n_q} - \ceil{m_{n_q}}/\phi$, whereas the latter one is bounded from above by
\[ \sum \limits_{i=1}^{n_q-1} \left( (\dt-1) \cdot \frac{\ceil{m_i}}{\phi} + \max \SET{ \frac{\ceil{m_i}}{\phi}, m_i } \right) + \dt \cdot \frac{\ceil{m_{n_q}}}{\phi} \DOT \]
Solution $(s', s'_{n_q})$ can use at most~$\dt$ objects of each group $b_{i,1}, \ldots, b_{i,d}$. Each of them, except one, can contribute at most $\frac{\ceil{m_i}}{\phi}$ to the $j^\text{th}$~total profit. One can contribute either at most~$\frac{\ceil{m_i}}{\phi}$ or at  most~$m_i$. This argument also holds for the $n_q^\text{th}$~group, but by the choice of index~$j$ we know that each object chosen by~$s'_{n_q}$ contributes at most~$\frac{\ceil{m_i}}{\phi}$ to the $j^\text{th}$~total profit. It is easy to see that $\ceil{m_i}/\phi \leq m_i$ because of $\phi > d \geq 1$. Hence, our bound simplifies to
\begin{align*}
\sum \limits_{i=1}^{n_q-1} &\left( (\dt-1) \cdot \frac{\ceil{m_i}}{\phi} + m_i \right) + \dt \cdot \frac{\ceil{m_{n_q}}}{\phi} \cr
&\leq \sum \limits_{i=1}^{n_q-1} \left( d \cdot \frac{m_i+1}{\phi} + m_i \right) + (d-1) \cdot \frac{m_{n_q}+1}{\phi} && \mbox{($d \geq 2$)} \cr
&= \frac{1}{\phi} \cdot \left( \sum \limits_{i=1}^{n_q-1} ( m_i \cdot (\phi+d) + d) + d \cdot (m_{n_q} + 1) \right) - \frac{m_{n_q}+1}{\phi} \cr
&= \frac{1}{\phi} \cdot \left( \sum \limits_{i=0}^{n_q-1} ( m_i \cdot (\phi+d) + d) + d \cdot m_{n_q} \right) - \frac{m_{n_q}+1}{\phi} && \mbox{($m_0 = 0$)} \cr
&= \frac{1}{\phi} \cdot ( (\phi-d) \cdot m_{n_q} + d \cdot m_{n_q}) - \frac{m_{n_q}+1}{\phi} && \mbox{(Equation~\eqref{eq.multi.recurrence})} \cr
&\leq m_{n_q} - \frac{\ceil{m_{n_q}}}{\phi} \DOT
\end{align*}
This implies that $(s', s'_{n_q})$ does not dominate $(s, s_{n_q})$.
\end{proof}
Immediately, we get a statement about the expected number of Pareto optimal solutions if we randomize.
\begin{corollary}
\label{multi.corol.pareto}
Let~$\S$ and~$b_{i,j}$ be as above, but the profit vectors~$q_{i,j}$ are arbitrarily drawn from~$Q_{i,j}$. Then, the expected number of Pareto optimal solutions for $K_\S(\{ b_{i,j} \WHERE i = \NLIST{n_q}, \ j = \NLIST{d} \})$ is at least $\left( \frac{2^d}{d} \right)^{n_q}$.
\end{corollary}
\begin{proof}
This result follows from Lemma~\ref{multi.lemma.copy} and the fact \[ |\hat{S}| = 1 + \binom{d}{\dt} \geq 1 + \frac{\sum \limits_{i=1}^d \binom{d}{i}}{d} = 1 + \frac{2^d-1}{d} \geq \frac{2^d}{d} \DOT \qedhere \]
\end{proof}
As in the bi-criteria case we now split each object~$b_{i,j}$ into~$k_i := \ceil{m_i}$ objects $b_{i,j}^{(1)}, \ldots, b_{i,j}^{(k_i)}$ with weights~$2^i/(k_i \cdot \dt)$ and with profit vectors
\[ q_{i,j}^{(l)} \in Q_{i,j}/k_i := \prod \limits_{k=1}^{j-1} \left[ 0, \frac{1}{\phi} \right] \times \left( \frac{m_i}{k_i} - \frac{1}{\phi}, \frac{m_i}{k_i} \right] \times \prod \limits_{k=j+1}^d \left[ 0, \frac{1}{\phi} \right] \DOT \]
Then, we adapt our set~$\S$ of solutions such that for any fixed indices~$i$ and~$j$ either all objects $b_{i,j}^{(1)}, \ldots, b_{i,j}^{(k_i)}$ are put into the knapsack or none of them. Corollary~\ref{multi.corol.pareto} yields the following result.
\begin{corollary}
\label{multi.corol.final.pareto}
Let~$\S$ and~$b_{i,j}^{(l)}$ be as described above, but the profit vectors $p_{i,j}^{(1)}, \ldots, p_{i,j}^{(k_i)}$ are chosen uniformly from~$Q_{i,j}/k_i$. Then, the expected number of Pareto optimal solutions of $K_\S(\{ b_{i,j}^{(l)} \WHERE i = \NLIST{n_q}, \ j = \NLIST{d}, \ l = \NLIST{k_i} \})$ is at least $\left( \frac{2^d}{d} \right)^{n_q}$.
\end{corollary}
Still, the lower bound is expressed in~$n_q$ and not in the number of objects used. So the next step is to analyze the number of objects.
\begin{lemma}
\label{multi.lemma.count.objects}
The number of objects~$b_{i,j}^{(l)}$ is upper bounded by $d \cdot n_q + \frac{2d^2}{\phi-d} \cdot \left( \frac{2\phi}{\phi-d} \right)^{n_q}$.
\end{lemma}
\begin{proof}
The number of objects~$b_{i,j}^{(l)}$ is $\sum_{i=1}^{n_q} (d \cdot k_i) = d \cdot \sum_{i=1}^{n_q} \ceil{m_i} \leq d \cdot n_q + d \cdot \sum_{i=1}^{n_q} m_i$, and
\begin{align*}
\sum \limits_{i=1}^{n_q} m_i &\leq \frac{d}{\phi+d} \cdot \sum \limits_{i=1}^{n_q} \left( \frac{2\phi}{\phi-d} \right)^i \leq \frac{d}{\phi+d} \cdot \frac{\left( \frac{2\phi}{\phi-d} \right)^{n_q+1}}{\left( \frac{2\phi}{\phi-d} \right) - 1}  \cr
&\leq \frac{d}{\phi} \cdot \left( \frac{2\phi}{\phi-d} \right) \cdot \left( \frac{2\phi}{\phi-d} \right)^{n_q} = \frac{2d}{\phi-d} \cdot \left( \frac{2\phi}{\phi-d} \right)^{n_q} \DOT \qedhere
\end{align*}
\end{proof}
Now we can prove Theorem~\ref{multi.mainthm}.
\begin{proof}[Proof of Theorem~\ref{multi.mainthm}]
Without loss of generality let~$n \geq 16d$ and~$\phi \geq 2d$. For the moment let us assume $\phi - d \leq \frac{4d^2}{n} \cdot \left( \frac{2\phi}{\phi-d} \right)^{\frac{n}{2d}}$. This is the interesting case leading to the first term in the minimum in Theorem~\ref{multi.mainthm}. We set $\hat{n}_q := \frac{\LOG{ (\phi-d) \cdot \frac{n}{4d^2} }}{\LOG{ \frac{2\phi}{\phi-d} }} \in \left[ 1, \frac{n}{2d} \right]$ and obtain $n_q := \floor{\hat{n}_q} \geq 1$ by rounding. All inequalities hold because of the bounds on~$n$ and~$\phi$. Now we consider objects~$b_{i,j}^{(l)}$, $i = \NLIST{n_q}$, $j = \NLIST{d}$, $l = \NLIST{k_i}$, with weights~$2^i/(k_i \cdot d)$ and profit vectors~$q_{i,j}$ chosen uniformly from~$Q_{i,j}/k_i$. All these intervals have length~$\frac{1}{\phi}$ and hence all densities are bounded by~$\phi$. Let~$N$ be the number of objects. By Lemma~\ref{multi.lemma.count.objects}, this number is bounded by
\begin{align*}
N &\leq d \cdot n_q + \frac{2d^2}{\phi-d} \cdot \left( \frac{2\phi}{\phi-d} \right)^{n_q} \leq d \cdot \hat{n}_q + \frac{2d^2}{\phi-d} \cdot \left( \frac{2\phi}{\phi-d} \right)^{\hat{n}_q} \cr
&\leq d \cdot \hat{n}_q + \frac{2d^2}{\phi-d} \cdot (\phi-d) \cdot \frac{n}{4d^2} \leq n \DOT
\end{align*}
Hence, the number~$N$ of binary variables we actually use is at most~$n$, as required. As set~$\S$ of solutions we use the set described above, encoding the copy step and the split step. Due to Corollary~\ref{multi.corol.final.pareto}, for fixed~$d \geq 2$ the expected number of Pareto optimal solutions of $K_\S(\{ b_{i,j}^{(l)} \WHERE i = \NLIST{n_q}, \ j = \SLIST{]}{d}, \ l = \NLIST{k_i} \})$ is
\begin{align*}
\OMEGA{ \left( \frac{2^d}{d} \right)^{n_q} } &= \OMEGA{ \left( \frac{2^d}{d} \right)^{\hat{n}_q} } = \OMEGA{ \left( \frac{2^d}{d} \right)^{\frac{\LOG{ (\phi-d) \cdot \frac{n}{4d^2} }}{\LOG{ \frac{2\phi}{\phi-d} }}} } = \OMEGA{ \left( (\phi-d) \cdot \frac{n}{4d^2} \right)^{\frac{\LOG{ \frac{2^d}{d} }}{\LOG{ \frac{2\phi}{\phi-d} }}} } \cr
&= \OMEGA{ (\phi \cdot n)^{\frac{d - \LOG[2]{d}}{\LOG[2]{ \frac{2\phi}{\phi-d} }}} } = \OMEGA{ (\phi \cdot n)^{(d - \LOG[2]{d}) \cdot (1 - \THETA{1/\phi})} } \KOMMA
\end{align*}
where the last step holds because of the same reason as in the proof of Theorem~\ref{bi.mainthm}.

In the case $\phi - d > \frac{4d^2}{n} \cdot \left( \frac{2\phi}{\phi-d} \right)^{\frac{n}{2d}}$ we construct the same instance above, but for a maximum density~$\phi' > d$ where $\phi' - d = \frac{4d^2}{n} \cdot \left( \frac{2\phi'}{\phi'-d} \right)^{\frac{n}{2d}}$. Since~$n \geq 16d$, the value~$\phi'$ exists, is unique and $\phi' \in [65d, \phi)$. Futhermore, we get $\hat{n}_q = \frac{n}{2d}$. As above, the expected size of the Pareto set is
\[ \OMEGA{ \left( \frac{2^d}{d} \right)^{\hat{n}_q} } = \OMEGA{ \left( \frac{2^d}{d} \right)^{\frac{n}{2d}} } = \OMEGA{ 2^{\THETA{n}} } \DOT \qedhere \]
\end{proof}

\bibliographystyle{alpha}
\bibliography{bibliography}

\begin{thebibliography}{BRV07}

\bibitem[BRV07]{DBLP:conf/ipco/BeierRV07}
Ren{\'e} Beier, Heiko R{\"o}glin, and Berthold V{\"o}cking.
\newblock The smoothed number of {P}areto optimal solutions in bicriteria
  integer optimization.
\newblock In {\em Proc. of the 12th Conference on Integer Programming and
  Combinatorial Optimization (IPCO)}, pages 53--67, 2007.

\bibitem[BV04]{DBLP:journals/jcss/BeierV04}
Ren{\'e} Beier and Berthold V{\"o}cking.
\newblock Random knapsack in expected polynomial time.
\newblock {\em Journal of Computer and System Sciences}, 69(3):306--329, 2004.

\bibitem[MO10]{MoitraO10}
Ankur Moitra and Ryan O'Donnell.
\newblock Pareto optimal solutions for smoothed analysts.
\newblock Technical report, CoRR (abs/1011.2249), 2010.
\newblock \url{http://arxiv.org/abs/1011.2249}.

\bibitem[RT09]{DBLP:conf/focs/RoglinT09}
Heiko R{\"o}glin and Shang-Hua Teng.
\newblock Smoothed analysis of multiobjective optimization.
\newblock In {\em Proc. of the 50th Ann. IEEE Symp. on Foundations of Computer
  Science (FOCS)}, pages 681--690, 2009.

\bibitem[ST04]{DBLP:journals/jacm/SpielmanT04}
Daniel~A. Spielman and Shang-Hua Teng.
\newblock Smoothed analysis of algorithms: Why the simplex algorithm usually
  takes polynomial time.
\newblock {\em Journal of the ACM}, 51(3):385--463, 2004.

\end{thebibliography}

\end{document}